\documentclass[11pt,draftclsnofoot,onecolumn]{IEEEtran}

\usepackage{amsfonts}
\usepackage{times}
\usepackage{latexsym}
\usepackage{amssymb}
\usepackage{amsmath}
\usepackage{cite}
\usepackage{verbatim}

\newcommand{\bydef}{\triangleq}
\newcommand{\tr}{{\it{tr}}}
\def\SNR{{\textsf{SNR}}}

\def\bydef{:=}

\def\bb0{{\mathbb{0}}}

\def\bydef{:=}

\def\bb{{\mathbf{b}}}

\def\bh{{\mathbf{h}}}

\def\bn{{\mathbf{n}}}

\def\bw{{\mathbf{w}}}
\def\bx{{\mathbf{x}}}
\def\by{{\mathbf{y}}}

\def\b0{{\mathbf{0}}}

\def\bA{{\mathbf{A}}}
\def\bB{{\mathbf{B}}}
\def\bC{{\mathbf{C}}}

\def\bH{{\mathbf{H}}}
\def\bI{{\mathbf{I}}}

\def\bbC{{\mathbb{C}}}

\def\bbE{{\mathbb{E}}}

\def\bbR{{\mathbb{R}}}

\def\bydef{:=}

\def\sf0{{\mathsf{0}}}

\def\Nt{{N_t}}
\def\Nr{{N_r}}

\def\nn{\nonumber}

\usepackage{graphicx}
\usepackage{amssymb}
\usepackage{amsfonts}
\usepackage{amsmath}
\usepackage{latexsym}
\usepackage{epstopdf}

\begin{document}

\newtheorem{thm}{Theorem}
\newtheorem{lemma}{Lemma}
\newtheorem{rem}{Remark}
\newtheorem{exm}{Example}
\newtheorem{prop}{Proposition}
\newtheorem{defn}{Definition}
\newtheorem{cor}{Corollary}
\def\proof{\noindent\hspace{0em}{\itshape Proof: }}
\def\endproof{\hspace*{\fill}~\QED\par\endtrivlist\unskip}
\def\bh{{\mathbf{h}}}
\def\SIR{{\mathsf{SIR}}}
\def\SINR{{\mathsf{SINR}}}

\title{Sub-modularity and Antenna Selection in MIMO systems}
\author{
Rahul~Vaze and Harish Ganapathy
\thanks{Rahul~Vaze is with the School of Technology and Computer Science, Tata Institute of Fundamental Research, Homi Bhabha Road, Mumbai 400005, vaze@tcs.tifr.res.in, Harish Ganapathy is with the Department of Electrical and Computer Engineering, The University of Texas at Austin, Austin, Tx, 78705, harishg@mail.utexas.edu. }
}

\date{}
\maketitle
\noindent
\abstract
In this paper, we show that the optimal receive antenna subset selection problem for maximizing the mutual information in a point-to-point MIMO system is sub-modular.  
Consequently, a greedy step-wise optimization approach, where at each step an antenna that maximizes the incremental gain is added to the existing antenna subset,  is guaranteed to be within a $(1-1/e)$ fraction of the global optimal value. For a single antenna equipped source and destination with multiple relays, we show that the relay antenna selection problem to maximize the mutual information is modular, when complete channel state information is available at the relays. As a result a greedy step-wise optimization approach leads to an optimal solution for the relay antenna selection problem with linear complexity in comparison to the brute force search that incurs exponential complexity.
\section{Introduction}

Transmit/receive antenna selection for point-to-point MIMO channels is a topic that has been extensively studied in literature, see\cite{Sandhu2000, Sanayei2004, Molisch2004,HeathAntSel2001,PaulrajAntSel2003,GershmanAntSel2004,ChenAntSel2010} and references therein.\footnote{The literature in this area is quite extensive and we do not provide a comprehensive list.} With transmit (receive) antenna selection, a subset of the total number of transmit (receive) antennas is chosen for maximizing 
various performance metrics, such as  capacity, reliability or diversity gain and several others. Antenna selection  has  numerous advantages, such as simplified circuitry, less number of transmit chains/power amplifiers etc., and therefore has been an object of interest in theory as well as in practice. With the growing popularity of relay based communication, antenna selection at both the transmitter/receiver, as well as at multiple relays has also
attracted a lot of attention \cite{Bletsas2006,Tanni2008, Peters2007, Vaze2008}. 
In the relay paradigm, in addition to usual advantages of antenna selection, relay antenna selection allows the set of relays to serve multiple source-destination 
pairs at the same time, thereby providing large spectral efficiency/reliability gains.

For point-to-point MIMO channels, assuming that antenna selection has been done (either by brute force, or a greedy method, or some simple heuristic approach), capacity expressions have been derived in \cite{Sanayei2004,Molisch2004}, while diversity gain computations have been provided in \cite{HughesAntSel}. Most of the analytical work in this area has concentrated on the evaluation of chosen metric given that the antenna selection has been done {\it apriori}. Finding the  optimal antenna subset, however, is a challenging problem on its own. 
For sufficiently large number of transmit/receive antennas (e.g. massive MIMO applications \cite{Massivemimo}), a brute force search is too expensive, and its compounded by the the fact that it needs to be done periodically for every coherence interval. 
There are a large number of papers on reduced complexity antenna selection algorithms \cite{Sandhu2000, PaulrajAntSel2003,GershmanAntSel2004,ChenAntSel2010}, however, most of them  do not provide any theoretical guarantees on the performance of the algorithm. The same holds true  for the relay antenna selection algorithms.

In this paper, in a major departure from the previous heuristic approaches, we study the antenna selection problem more systematically by leveraging results from  the area of
approximation algorithms. Approximation algorithms for solving combinatorial optimization problems is a major field of study in computer science \cite{BookVazirani}, 
where an approximate solution to an optimization problem is derived that has a fixed bounded distance from the optimal solution. 
One of the techniques used in approximation algorithms is to check if the objective function is sub-modular, since in discrete combinatorial optimization, sub-modular objective functions play a role that is akin to convex functions in the continuous domain. Sub-modular functions have been a topic of study even since the celebrated result of \cite{Nemhauser1978}, that showed that a greedy algorithm (maximize per step reward) achieves a $(1-1/e)$ fraction of the optimal solution \cite{Nemhauser1978} if the objective function is sub-modular. A function is called sub-modular  if it satisfies a diminishing
returns property, i.e. the marginal gain from adding an element
to a set $S$ is at least as high as the marginal gain from adding the same element to a superset of $S$. A special case of a sub-modular function is a modular function (non-diminishing return) for which the value from adding an element $a$
to a set $S$ is equal to the sum of the value just using $S$ and the value with using $a$. For a modular function, 
it is well-known that a greedy algorithm achieves the optimal solution \cite{RadoGreedyOpt1968, DavidGreedyOpt1968,EdmondsGreedyOpt1971}.

In this paper, for point-to-point MIMO channels, we study the receive antenna selection problem for maximizing the mutual information or achievable rate, where the goal is to select the $L$ best antennas among the total $\Nr$ receive antennas. We assume that the number of transmit antennas $\Nt \le \Nr$. 
We  show that the objective function in the receive antenna selection problem is sub-modular, and hence the mutual information with the greedy algorithm 
is guaranteed to be within a $(1-1/e)$ fraction of the optimal mutual information value. The greedy algorithm at each step updates the transmit antenna subset by adding that antenna to the existing subset that has the highest increment to the mutual information among the available antennas. Therefore, the  
complexity of the greedy algorithm is linear in the number of antennas. Thus, the greedy antenna selection policy not only has guaranteed performance bound but is also computationally simple for practical implementation.

For the relay antenna selection problem we consider a single antenna equipped source, several relays with total $N$ antennas, and a destination with a single antenna. The problem we consider is to select $L$ best antennas out of the available $N$ antennas to maximize the mutual information between the source and the destination.  
For the relay antenna selection problem we show that the objective function of the relay antenna selection problem is modular, i.e. the objective function with using $n$ relay antennas is equal to the sum of objective function with using only $n-1$ antennas and the objective function with using only a single antenna. Hence, using the results from approximation algorithms, we show that a greedy algorithm achieves the optimal solution, however, with linear complexity compared to the exponential complexity of the brute force approach.

\section{Notation}
Let ${\bA}$ denote a matrix, ${\bf a}$ a vector and
$a_i$ the $i^{th}$ element of ${\bf a}$. The transpose and conjugate transpose are denoted by $^T$, and $^\dag$, respectively. ${\bf I}_n$ denotes the $n\times n$ identity matrix.
The expectation of function $f(x)$ with respect to $x$ is denoted by
${\bbE}(f(x))$.
Let $S_1$ be a set and $S_2$ be a subset of $S_1$. Then $S_2 \backslash S_1$ denotes the set of elements of $S_1$ that do not belong to $S_2$. The cardinality of set $S$ is denoted by $|S|$. We use the symbol
$\bydef$  to define a variable.
 
\section{Organization}The rest of the paper is organized as follows. Section \ref{sec:sys} describes the system model and problem statement for the receive antenna selection problem in point-to-point MIMO channels. In Section \ref{sec:p2p}, 
we establish the sub-modularity of the antenna receive antenna selection problem in point-to-point MIMO channels. 
The relay selection problem is discussed and analyzed in Section \ref{sec:relay}, where the modularity of relay selection problem is derived. Section \ref{sec:sims} illustrates some numerical examples, and some final conclusions are made in Section \ref{sec:conc}.

\section{System Model for the Point-to-Point MIMO channel}\label{sec:sys}
Consider a MIMO wireless channel between a transmitter  with $\Nt$ antennas and its receiver with $\Nr \ge \Nt$ antennas. Assuming that the receiver 
uses $L \le \Nr$ antennas (indexed by ${\cal R}_L$) to receive the  $L$ independent streams $\bx = [x_1,\dots, x_{L}]^T$, the received signal 
at the destination is given by
\begin{equation}\label{eq:rxsig}
\by = \sqrt{\frac{P}{\Nt}} \bH_{{\cal R}_L} \bx + \bn,
\end{equation}
where $P$ is the average transmit power, $\bH_{{\cal R}_L}\in \bbC^{\Nr\times L}$ is the channel coefficient matrix with entries $\bH_{{\cal R}_L}(i,j)$ corresponding to 
the channel coefficient between the $i^{th}$ receive antenna of the selected set ${\cal R}_L$ and the $j^{th}$ transmit antenna. The results of this paper are applicable for any continuous distribution on the entries of $\bH_{{\cal R}_L}$. We assume that the receiver exactly knows the channel state information (CSI) $\bH_{{\cal R}_L}, \forall \ {\cal R}_L \subset \{1,\dots,\Nr\} $  and uses it for 
performing antenna selection, while the transmitter has no CSI and uses all its antennas with equal power allocation. 
The case with CSI available at the transmitter while of equal interest is presently out of scope of this paper.

There are several reasons for using only $L$ antennas out of the total available $\Nr$,  such as limited number of actual transmitter chains/power amplifiers, simple circuitry, etc. See \cite{Molisch2004,Sanayei2004} for a detailed discussion. In this paper we consider the selection of $L$ best receive antennas for maximizing mutual information, which is key for maximizing achievable rate (capacity) and diversity-multiplexing tradeoff through outage capacity maximization.  


\begin{rem}\label{rem:nonsubmodtransmit} Transmit antenna selection problem where $L$ out of $\Nt$ antennas are selected at the transmitter for maximizing the mutual information is not sub-modular. It is actually not even monotonically increasing in the number of antennas, and hence a greedy algorithm cannot be guaranteed to give theoretical guarantees on its performance. To illustrate the non-monotonicity consider the following example, where there are $\Nt=2$ transmit antennas and a single receive antenna $\Nr=1$. Let the channel magnitudes between the $i^{th}$ transmit antenna and the receiver antenna be  $h_i, \ i=1,2$. Then if only one antenna is used, say the first, then the mutual information is $C_1 =\log (1+P |h_1|^2)$, while if both antennas are used $C_2 = \log (1+\frac{P}{2} (|h_1|^2+|h_2|^2))$. Now depending on values of $h_1$ and $h_2$, $C_1$ can be more or less than $C_2$. Thus, the transmit antenna selection problem is not monotonically increasing in the number of transmit antennas. This limitation arises because of equal power splitting among the transmit antennas. Thus, increasing the number of transmit antennas does not increase the mutual information for each realization of channel magnitudes, however, in expectation (ergodic capacity) increasing the number of antennas does help. 
\end{rem}

\section{Sub-Modularity of the Receive Antenna Selection in a Point-to-Point MIMO channel}\label{sec:p2p}
For the receive signal model (\ref{eq:rxsig}),  the mutual information  using antenna subset ${\cal R}_L$  \cite{Cover2004}, is given by 
\begin{equation}\label{eq:mi}
C_{{\cal R}_L} \bydef  \log \det \left(\bI + \frac{P}{\Nt}\bH_{{\cal R}_L}\bH_{{\cal R}_L}^{\dag}\right),
\end{equation}
and the antenna selection problem is to find the optimal set ${\cal R}_L$ of $L$ transmit antennas that maximizes the mutual information, i.e.
\begin{equation}\label{optprob}
\max_{{\cal R}_L \subset \{1,2,\dots,\Nr\}, |{\cal R}_L|=L} C_{{\cal R}_L}. 
\end{equation}

\begin{rem} For MIMO channels, the tradeoff between rate of transmission and reliability (diversity gain) is captured through the diversity-multiplexing  tradeoff (DMT) \cite{Zheng2003}. The DMT directly depends on the outage probability 
that is defined as \[P_{out}(R) \bydef P(C_{{\cal R}_L} \le R).\] Thus, maximizing $C_{{\cal R}_L}$ over antenna subsets is equal to obtaining the optimal  diversity-multiplexing tradeoff of MIMO channels with antenna selection.
\end{rem}

The antenna selection problem (\ref{optprob}) has no elegant solution, and with a brute force approach, one needs to make $ {\Nr \choose L}$ computations 
to solve it. Since $\Nr$ can be large, e.g. in order of tens or hundreds for massive MIMO applications \cite{Massivemimo}, the brute force search is too expensive. Moreover, 
the search needs to be carried out after each coherence interval since channel coefficients change independently across different coherence intervals. To simplify the complexity of antenna selection, many heuristic approaches have been proposed in literature, however, none of them provide any theoretical guarantees on their performance. 

In this paper, we present a systematic study of receive antenna selection algorithm and provide theoretical guarantees on its performance by leveraging ideas from the field of approximation algorithms that are quite popular  in the computer science community. Essentially, we make use of the well-known result that for a sub-modular objective function, a greedy solution approximates 
the optimal solution by a factor of $(1-1/e)$ \cite{Nemhauser1978}. 
We will show that the receive antenna selection problem is sub-modular in the number of antennas, 
and hence a greedy solution, which at each step maximizes the incremental gain in the mutual information achieves a $(1-1/e)$ fraction of the 
optimal solution. To begin with, we need the following definitions.

\begin{defn} Let $f$ be a function defined as $f : U \rightarrow \bbR^+$. Then $f$ is called {\it monotone} if $f(S \cup \{a\}) - f(S) \ge 0$, for all $a\in U, S\subseteq U, a\notin S$, and $f$ is called a {\it sub-modular} function if it satisfies 
\[f(S \cup \{a\}) - f(S) \ge f(T \cup \{a\}) - f(T), \] for all elements $a\in U, a\notin T$ and all pairs of subsets $S \subseteq T \subseteq U$. In particular, a function $f$ is called {\it modular} if it satisfies \[f(S \cup \{a\}) - f(S) = f(T \cup \{a\}) - f(T), \] for all elements $a\in U, a\notin T$ and all pairs of subsets $S \subseteq T \subseteq U$.
\end{defn}
Essentially, for a sub-modular function the incremental gain from adding an extra element in the set decreases with the size of the set. The main interest in sub-modular functions is because of following Theorem that provides guarantees on the performance of greedy methods for optimizing sub-modular objective functions.

\begin{thm} \cite{Nemhauser1978}\label{thm:approx} For a non-negative, monotone sub-modular
function $f$, let $S$ be a set of size $k$ obtained by selecting elements
one at a time, each time choosing an element that provides
the largest marginal increase in the function value. Let $S^{\star}$ be a
set that maximizes the value of $f$ over all $k$-element sets. Then
$f(S)\ge (1-\frac{1}{e})f(S^{\star})$, in other words, $S$ provides a $(1-\frac{1}{e})$
approximation.
\end{thm}

For the special case of modular functions, a stronger result is available that is as follows.
\begin{thm} \cite{RadoGreedyOpt1968, DavidGreedyOpt1968,EdmondsGreedyOpt1971}\label{thm:greedyoptmodular} For a non-negative, monotone modular
function $f$, let $S$ be a set of size $k$ obtained by selecting elements
one at a time, each time choosing an element that provides
the largest marginal increase in the function value. Let $S^{\star}$ be a
set that maximizes the value of $f$ over all $k$-element sets. Then
$f(S)=f(S^{\star})$, in other words, the subset $S$ obtained by the greedy method is optimal.
\end{thm}

Thus, if we can show that the receive antenna selection problem (\ref{optprob}) is sub-modular, then we are guaranteed to get a $(1-\frac{1}{e})$ 
approximation using a greedy step-wise approach. Next, we provide two different proofs for showing that the receive antenna selection problem objective function is sub-modular. Before the proof we need the following Lemma.

\begin{lemma}\label{lem:entropysubmod} Consider an $N$ dimensional random vector $\bx$. Let $\bx_A=[\bx_i]_{i\in A}$ be the vector consisting 
of elements of $\bx$ indexed by $A\subseteq \{1,2,\dots,N\}$, and let $h({\bx_A})$ be the entropy of the random vector $\bx_A$ \cite{Cover2004}. 
Then the entropy function is sub-modular over the subsets of $\{1,2,\dots,N\}$, i.e. 
$h(\bx_{S\cup \{a\}}) - h(\bx_{S}) \ge h(\bx_{T\cup \{a\}}) - h(\bx_{T})$ when $S\subseteq T \subseteq \{1,2,\dots,N\}, a \in \{1,2,\dots,N\}, a \notin T$.
\end{lemma}

\begin{proof} By the definition of entropy \cite{Cover2004}, $h(\bx_{S\cup \{a\}}) - h(\bx_{S}) = h(\bx_{\{a\}}|\bx_S)$. Since $S\subseteq T$, clearly, the conditional entropy of $\bx_{\{a\}}$ given $\bx_T$ is less than given $\bx_S$, i.e. $h(\bx_{\{a\}}|\bx_S) \ge h(\bx_{\{a\}}|\bx_T)$.
\end{proof}

\begin{thm}\label{thm:monotone} The objective function $C_{{\cal R}_L}$ is monotone over the subsets of $\{1,2,\dots,\Nr\}$.
\end{thm}
\begin{proof} The proof is immediate since adding more receive antennas cannot decrease the mutual information. It can also be shown more directly using determinant inequalities.
\end{proof}

\begin{thm}\label{thm:submod} The objective function $C_{{\cal R}_L}$ is sub-modular over the subsets of $\{1,2,\dots,\Nr\}$.
\end{thm}
\begin{proof} In this proof we will show that $C_{{\cal R}_L}$ is equal to the entropy of some random variable, and conclude the proof 
using the fact that entropy function is sub-modular using Lemma \ref{lem:entropysubmod}. 
Recall from (\ref{eq:mi}) that 
\[C_{{\cal R}_L} =   \log \det \left(\bI + \frac{P}{\Nt}\bH_{{\cal R}_L}\bH_{{\cal R}_L}^{\dag}\right).\] 
Let $\Sigma \bydef  \left(\bI + \frac{P}{\Nt}\bH_{{\cal R}_L}\bH_{{\cal R}_L}^{\dag}\right)$. 
Consider a set $V = \{1,2,\dots, \Nr\}$, and let $A \subseteq V$ of cardinality $|A| = L$. Let $\bx_A$ be a zero-mean, multi-variate 
Gaussian random vector with covariance matrix $\Sigma$, i.e. the probability density function of $\bx_A$ is 
$g_{\bx_A}(x) = \frac{1}{\sqrt{2\pi \det(\Sigma)}} \exp^{-\frac{1}{2}x^T\Sigma^{-1}x}$. This construction is valid since $\Sigma$ is a symmetric matrix. Also note that $\Sigma$ is a positive definite matrix with $\det(\Sigma) > 1$, and hence $\log(\det(\Sigma)) >0$. Then, from \cite{Cover2004},
 the entropy of $\bx_A$ is $h(\bx_A) = \log \left(\sqrt{ (2\pi e)^L \det(\Sigma)}\right)$. 
 Hence 
 \begin{eqnarray*}
h(\bx_A) &=& \log \left(\sqrt{2\pi e^L}\right) + \frac{1}{2} \log \left(\det(\Sigma)\right), \\ 
&=& \log \left(\sqrt{2\pi e^L}\right) + C_{{\cal R}_L}. 
\end{eqnarray*}
Since entropy is a sub-modular function (Lemma \ref{lem:entropysubmod}), it follows that  $C_{{\cal R}_L}$ is sub-modular.

\end{proof}

We present another proof of Theorem \ref{thm:submod} to illustrate the connections between sub-modularity of $C_{{\cal R}_L}$ and mutual information of some information theoretic channels.

\begin{proof}
In this proof we will directly show that for $f({\cal R}_L) = C_{{\cal R}_L}, \ {\cal R}_L\subseteq U = \{1,2,\dots,\Nr\}$, where $U$ denotes the set of receive antennas, $f(S \cup \{a\}) - f(S) \ge f(T \cup \{a\}) - f(T)$, for all elements $a\in U, a \notin T,$ and all pairs of subsets $S \subseteq T \subseteq U$. 

For $U = \{1,2,\dots,\Nr\}$, and $a\in U, a \notin T, S \subseteq T \subseteq U$,
let the channel coefficient matrix between the $\Nt$ transmit antennas and the set $S$ of receiver antennas  be $\bH_S\in \bbC^{|S|\times \Nt}$, between the  $\Nt$ transmit antennas and the set $T$ of receiver antennas  be $\bH_T\in \bbC^{ |T| \times \Nt} = [\bH_{T\backslash S} \ \bH_S]^T$, and the channel coefficient vector between the $\Nt$ transmit antennas and the $a^{th}$ receive antenna be $\bh\in \bbC^{1\times \Nt}$.
Then for $f({\cal R}_L) = C_{{\cal R}_L}$,
\begin{eqnarray}\nn
f(S \cup \{a\}) - f(S) &=&  \log \det \left(\bI_{|S|+1} + \frac{P}{\Nt}\left[\begin{array}{c}\bH_S \\  \bh\end{array}\right][\bH_S^{\dag} \ \ \bh^{\dag}]\right) -  \log \det \left(\bI_{|S|} + \frac{P}{\Nt}\bH_S \ \bH_S^{\dag}\right),\\ \label{eq:mi1} 
 &=& \log \det \left(\bI_{\Nt} + \frac{P}{\Nt}[\bH_S^{\dag} \ \ \bh^{\dag}]\left[\begin{array}{c}\bH_S \\  \bh\end{array}\right]\right) -  \log \det \left(\bI_{\Nt} + \frac{P}{\Nt}\bH_S^{\dag} \ \bH_S\right), 
\end{eqnarray}
where the second statement follows from the determinant equality $\det(I + AB) = \det(I+BA)$. Similarly 
\begin{eqnarray}\nn
f(T \cup \{a\}) - f(T) &=&  \log \det \left(\bI_{\Nt} + \frac{P}{\Nt} [\bH_{T\backslash S}^{\dag} \ \bH_S^{\dag} \ \bh^{\dag}]\left[\begin{array}{c}\bH_{T\backslash S} \\  \bH_S \\ \bh\end{array}\right]\right)\\\label{eq:mi2}
&& -  \log \det \left(\bI_{\Nt} + \frac{P}{\Nt}[\bH_{T\backslash S}^{\dag} \ \bH_S^{\dag}] \left[\begin{array}{c}\bH_{T\backslash S} \\  \bH_S \end{array}\right]\right).
\end{eqnarray}

Consider the two MIMO multiple access channels (MAC) shown in Figs. \ref{fig1} and \ref{fig2}. In Fig. \ref{fig1}, a two-user MAC is shown, where $A$ is a user with a single antenna, $B$ is a user with $|S|$ antennas, and the receiver $C$ has $\Nt$ antennas. Similarly, in Fig. \ref{fig2}, a three-user MAC is shown, where  
$X$ is a user with a single antenna, $Y$ is a user with $|S|$ antennas, $Z$ is a user with $| T \backslash S|$ antennas, and the receiver $W$ has 
$\Nt$ antennas.
Assuming all input distributions are Gaussian in 
Figs. \ref{fig1} and \ref{fig2}, and each antenna of each user is transmitting power $\frac{P}{\Nt}$, it is easy to see that (\ref{eq:mi1}) corresponds to the mutual information between $A$ and $C$ in Fig. \ref{fig1}, while (\ref{eq:mi2}) corresponds to to the mutual information between 
$X$ and $W$ in Fig. \ref{fig2}. Since, the channel in Fig. \ref{fig2} between $X$ and $W$ is physically degraded \cite{Cover2004} compared to the channel between $A$ and $C$ in Fig. \ref{fig1}, the mutual information of the channel in Fig. \ref{fig1} is at least as much as the mutual information of the channel in Fig. \ref{fig2}, and hence we can conclude that 
expression (\ref{eq:mi1}) is greater than or equal to (\ref{eq:mi2}), consequently proving the sub-modularity of the objective function in the 
antenna selection problem.

\end{proof}

\begin{figure}
\centering
\includegraphics[width=2in]{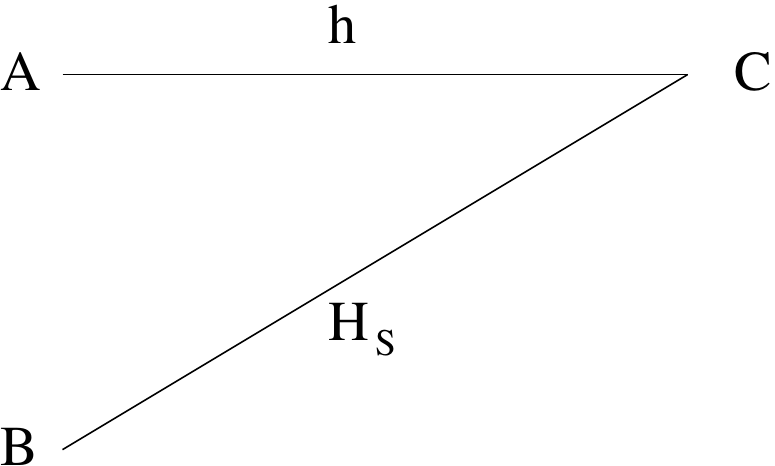}
\caption{MIMO multiple access channel  with two inputs and one output.}
\label{fig1}
\end{figure}

\begin{figure}
\centering
\includegraphics[width=2in]{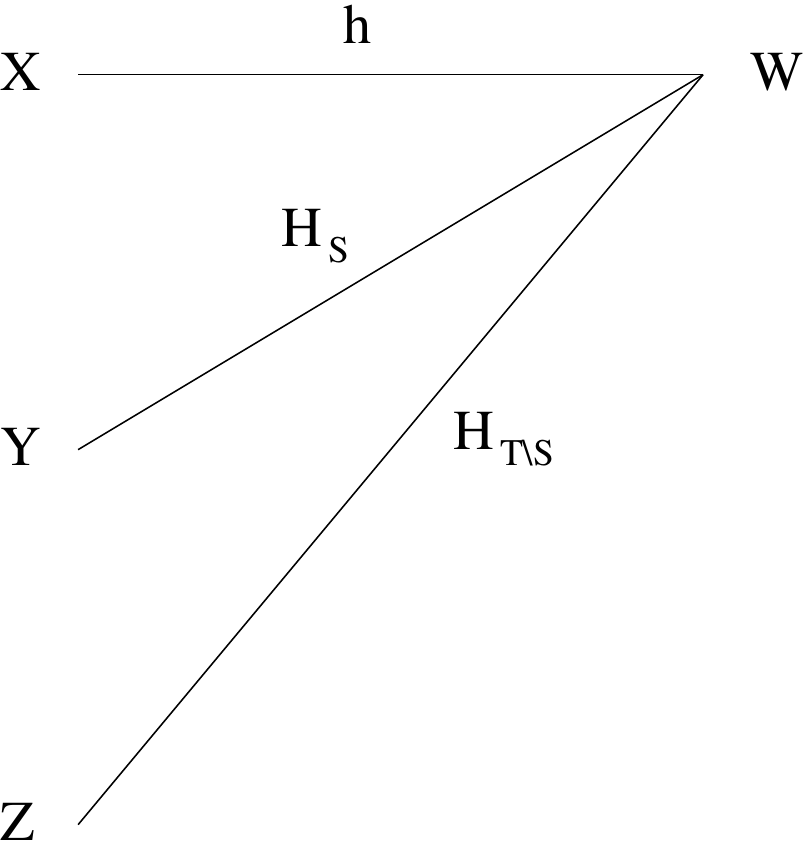}
\caption{MIMO multiple access channel  with three inputs and one output.}
\label{fig2}
\end{figure}

Thus, in light of Theorem \ref{thm:approx}, and the sub-modularity of the objective function in the receive antenna selection problem (Theorem \ref{thm:submod}), we propose the following 
greedy algorithm to maximize the mutual information while selecting the $L$ out of the $\Nr$ receive antennas.

Greedy Algorithm (GA):
At step $i$, ${\cal R}_L = {\cal R}_L \cup \{i^{\star}\}$, where  \[i^{\star} = \arg \max_{i \in \{1,2,\dots, \Nr\}, i\notin {\cal R}_L}  \log \det \left(\bI + \frac{P}{\Nt}\bH_{{\cal R}_L\cup \{i\}}\bH_{{\cal R}_L\cup \{i\}}^{\dag}\right),\] and  
repeat for $i=i+1$. Stop when $|{\cal R}_L|=L$.

Formally, we write our main result as follows. 
\begin{thm} \label{thm:} Let the output of the greedy algorithm (GA) be the set $S$. Then, if $S^{\star}$ is the
set that maximizes the value of $f({\cal R}_L) =C_{{\cal R}_L}$ over all $L$-element sets. Then
$f(S)\ge (1-\frac{1}{e})f(S^{\star})$.
\end{thm}
\begin{proof} Since $\det \left(\bI + \frac{P}{\Nt}\bH_{{\cal R}_L}\bH_{{\cal R}_L}^{\dag}\right) >1$,  $C_{{\cal R}_L} > 0$ and it follows that  $f = C_{{\cal R}_L}$ is a function from  $\{1,2,\dots,\Nr\}$ to $\bbR^+$. Thus, the result  
follows from Theorem \ref{thm:approx}, using the monotonicity of the receive antenna selection problem from Theorem \ref{thm:monotone}, and sub-modularity of the receive antenna selection problem from Theorem \ref{thm:submod}.
\end{proof}

{\it Discussion:} In this section, we showed that in a point-to-point MIMO channel, the receive antenna selection problem is sub-modular, and a greedy optimization approach is guaranteed to be within $(1-1/e)$ fraction of the optimal solution. While the transmit/receive antenna selection problem has received tremendous amount of attention and  is well studied in literature, however, to the best of our knowledge, no such theoretical guarantees have been proven before this work. 
Prior work on finding the optimal antenna subset primarily uses heuristic low-complexity approaches and does not provide with theoretical bounds. By making use of results available in the approximation algorithms literature, we have been able to justify the use of greedy approaches for receive antenna selection and derive a lower bound on their performance.

\section{Relay Antenna Selection}\label{sec:relay}
In this section, we consider a multiple relay network, where multiple relays with total $N$ antennas (co-located or distributed does not matter) help the communication between the source and the destination, equipped with single antenna each.
We show that the relay selection problem in a multiple relay network is modular when each relay has full CSIT for both of its channels and uses an amplify and forward strategy. Using the modularity of the relay selection problem, then we conclude that a greedy optimization approach leads to the optimal solution for the relay selection problem.

The relay selection problem is invariant to the presence of direct path or any duplexity assumption. So for simplicity we assume that there is no direct path and relays work in half-duplex mode with a sum power constraint. Hence the transmission takes place in two phases, where in first phase the source transmits to all relays, and then in the next phase  the selected relays transmit to the destination. Without loss of generality, for simplicity we assume that the source transmits with unit average power, and there is an unit average sum power constraint on the relays. Let the channel between the source and the $i^{th}$ relay be $f_i$ and the channel between the $j^{th}$ relay and 
destination be $g_j$. We assume that the source has no CSI, while the $i^{th}$ relay has CSI for both its channel coefficients $f_i$ and $g_i$, and the destination is assumed to know all the channel coefficients $f_i, g_i, \forall \ i$.

Let $y_k$ be the signal received at relay $k$, where 
\begin{equation*}
y_k = f_k x + n_k,
\end{equation*}
where $x$ is the signal transmitted by the source with $\bbE(x^2) =1$, and $n_k$ is the AWGN with zero mean and unit variance. The relay $k$ then 
transmits $t_k = \frac{w_ky_k}{\gamma_k}$, where $\gamma_k = \sqrt{|f_k|^2+1}$ is the normalization factor to ensure that $\bbE\left(\left(\frac{y_k}{\gamma}\right)^2\right)=1$, and $\sum_{k=1}^N w_kw_k^{\dag} = 1$ to ensure a unit  sum power constraint.

Under these assumptions, if ${\cal T}_L\subseteq \{1,2,\dots,N\}$ is the subset of relays chosen for transmission, the received signal at the destination is 
\[\by = \sum_{i\in {\cal T}_L}  \frac{g_i w_i f_i}{\gamma_i} x +  \sum_{i\in {\cal T}_L}  \frac{g_iw_i}{\gamma_i} n_i + v,\] where $v$ is the AWGN with zero mean and unit variance received at the destination. Let ${\cal T}_L = \{t_1, \dots, t_L\}$.
Hence the mutual information with the relay selected set ${\cal T}_L$ is 
\begin{equation}
C_{{\cal T}_L}^{relay}(\bw) = \log \left(1+ \SNR_{{\cal T}_L}\right),
\end{equation}
where
\[\SNR_{{\cal T}_L}(\bw) \bydef \frac{\bw^{\dag} \Delta \Delta^{\dag}\bw}{\bw^{\dag} ( \Sigma\Sigma^{\dag}+\bI)\bw},\] 
with
$\bw = [w_{t_1}, \dots, w_{t_L}]^T$, $\Delta = \left[ \frac{g_{t_1}f_{t_1}}{\gamma_{t_1}}, \dots,  \frac{g_{t_L} f_{t_L}}{\gamma_{t_L}}\right]^T$, and $\Sigma = \left[\begin{array}{ccc}  \frac{g_{t_1}}{\gamma_{t_1}} & 0 & 0 \\ 
0& \ddots &0
\\
 0&  0&\frac{g_{t_L}}{\gamma_{t_L}} \end{array}\right]$, since $\bw^{\dag}\bw =1$.

From the relay selection point of view, the optimization problem is 
\begin{equation*}
\max_{ {\cal T}_L \subseteq \{1,2,\dots,N\}} \max _{\bw} \ C_{{\cal T}_L}^{relay}(\bw),
\end{equation*}
or equivalently 
\begin{equation}
\label{optprobrelay}
\max_{{\cal T}_L \subseteq \{1,2,\dots,N\}} \max_{\bw}\SNR_{{\cal T}_L}(\bw). 
\end{equation}

Next, we show that the relay selection problem (\ref{optprobrelay}) is modular and thereafter from Theorem \ref{thm:greedyoptmodular}, we conclude that the greedy solution to the relay antenna selection problem is optimal.

\begin{thm} The objective function of the optimal relay selection problem (\ref{optprobrelay}) is modular.\label{thm:relaymod} 
\end{thm}  
\begin{proof} Consider $\SNR_{{\cal T}_L}(\bw) \bydef \frac{\bw^{\dag} \Delta \Delta^{\dag}\bw}{\bw^{\dag} ( \Sigma\Sigma^{\dag}+\bI)\bw}$. Let $\bA = \Delta \Delta^{\dag}$ and $\bB = \Sigma\Sigma^{\dag}+\bI$. Then 
\begin{eqnarray*}
\SNR_{{\cal T}_L}(\bw) &=& \frac{\bw^{\dag}\bA\bw}{\bw^{\dag}\bB\bw},\\
&=& \frac{\bw^{\dag}\bB^{1/2}(\bB^{\dag})^{-1/2}\bA \bB^{-1/2}(\bB^{\dag})^{1/2}\bw}{\bw^{\dag}\bB^{1/2}(\bB^{\dag})^{1/2}\bw}, \ \text{since $\bB$ is positive definite and symmetric},\\
&=& \frac{\by^{\dag}(\bB^{\dag})^{-1/2}\bA \bB^{-1/2}\by}{\by^{\dag}\by},  \ \ \ \ \ \by \bydef (\bB^{\dag})^{1/2}\bw, \\
 &=& \frac{\by^{\dag}\bC\by}{\by^{\dag}\by},  \ \ \ \ \ \bC \bydef  (\bB^{\dag})^{-1/2}\bA \bB^{-1/2}. \\
\end{eqnarray*}
Thus, $\max_{\bw} \SNR_{{\cal T}_L}(\bw) = \max_{\by}\frac{\by^{\dag}\bC\by}{\by^{\dag}\by} = \lambda_{max}(\bC)$ from the 
Rayleigh-Ritz theorem \cite{Horn1985}, and the optimal $\by$ is the eigen-vector of $\bC$ 
corresponding to the largest eigen-value of $\bC$. Moreover, since $\bA$ is a rank-$1$ matrix, $\bC$ is also rank-$1$, with  $\lambda_{max}(\bC) = \tr(\bC)$, and the optimal $\by =\bB^{-1/2} \Delta$, and consequently the optimal $\bw =  \bB^{-1} \Delta$. Since $\tr(\bC) = \sum_{i \in {\cal T}_L} \frac{|g_i|^2|f_i|^2}{|f_i|^2+|g_i|^2+1} $, we have that  
\[\max_{\bw} \SNR_{{\cal T}_L}(\bw) =  \sum_{i \in {\cal T}_L} \frac{|g_i|^2|f_i|^2}{|f_i|^2 + |g_i|^2+1}.  \] Thus, clearly the objective function is modular in the number of relay antennas.
\end{proof}

The greedy algorithm to maximize the capacity while selecting $L$ relay antennas out of $N$ is as follows.
Greedy Algorithm for Relay Selection (GARS): Initialize $n=1$ and ${\cal T}_L =\phi$.
At step $n$, ${\cal T}_L = {\cal T}_L \cup \{i^{\star}\}$, where  \[i^{\star} = \arg \max_{i \in \{1,2,\dots, N\}, i\notin {\cal T}_L}  \SNR_{{\cal T}_L \cup \{i\}},\] and repeat for $n=n+1$. Stop when $|{\cal T}_L|=L$.

The main result of this section is as follows. 
\begin{thm} \label{thm:greedyoptrelay} Let the output of the greedy algorithm (GARS) be the set $S$. 
Then, if $S^{\star}$ is the
set that maximizes the value of $f=C_{{\cal T}_L}$ over all $L$-element sets. Then
$f(S) = f(S^{\star})$.
\end{thm}
\begin{proof} Since the relay selection problem is modular (Theorem \ref{thm:relaymod}), the result follows from Theorem  \ref{thm:greedyoptmodular}.
\end{proof}

{\it Discussion:} In this section, we showed that in a multiple relay network with a single antenna equipped source-destination pair, the relay antenna selection problem for maximizing the mutual information is modular. Thus, a greedy optimization approach  achieves the optimal solution. The modularity of the objective function in this case follows by using matrix theory results to show that each relay antenna contributes an additive term to the objective function. Modular functions are special functions that do not exhibit the diminishing returns property, and where the incremental gain of adding a new element to an existing set is identical no matter how large the existing set is. 
Relay antenna selection has been extensively studied in the literature \cite{Caleb2007,Bletsas2006,Ibrahim2008,JingAntSel2009,Zhao2007}, with various objective functions, but to the best of our knowledge this is the first work that derives theoretical guarantees on relay antenna selection algorithms, let alone the optimality of the greedy approach.

\section{Simulations}\label{sec:sims} In this section we present some numerical results to illustrate the results derived in this paper. 
In Fig. \ref{sim:figp2p}, we first consider the point-to-point MIMO case and plot the achievable rate (mutual information) versus the number of chosen receive antennas for both the greedy as well the optimal strategy (brute-force) for $\Nt=4, \Nr=16$ with unit power transmission. As can be seen from Fig. \ref{sim:figp2p}, the performance of the greedy algorithm is almost similar to the optimal strategy, and far better than the theoretically derived result (being $(1-1/e)$ fraction of the optimal). The point to note here is that the theoretical bound is somewhat pessimistic and corresponds to the worst-case scenario, however, in most cases the performance of greedy algorithms is significantly better than the promised 
worst-case bound.
Similar simulations results have been obtained in \cite{GershmanAntSel2004} to show that greedy algorithms almost achieve the optimal performance in antenna selection for point-to-point MIMO setting. In Fig. \ref{fig:simrelay}, we consider the relay network and consider $N=16$ relay antennas to select from, and plot the achievable rate (mutual information) versus the number of chosen relay antennas for both the greedy as well the optimal strategy (brute-force). As established in this paper, the greedy algorithm achieves the optimal performance.

\begin{figure}
\centering
\includegraphics[width=4in]{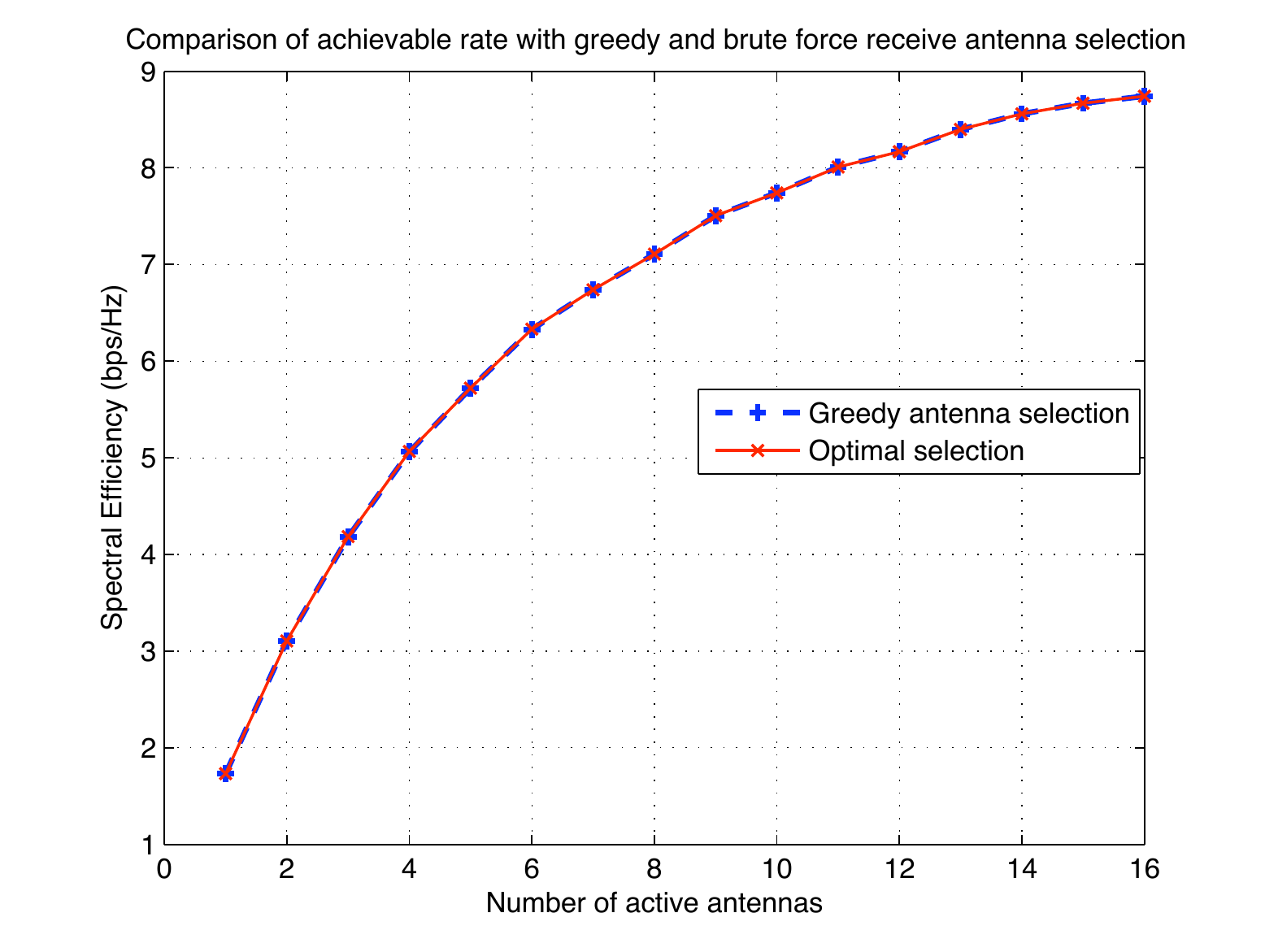}
\caption{Comparison of greedy v/s optimal strategy for receive antenna selection in a point-to-point MIMO channel.}
\label{sim:figp2p}
\end{figure}

\begin{figure}
\centering
\includegraphics[width=4in]{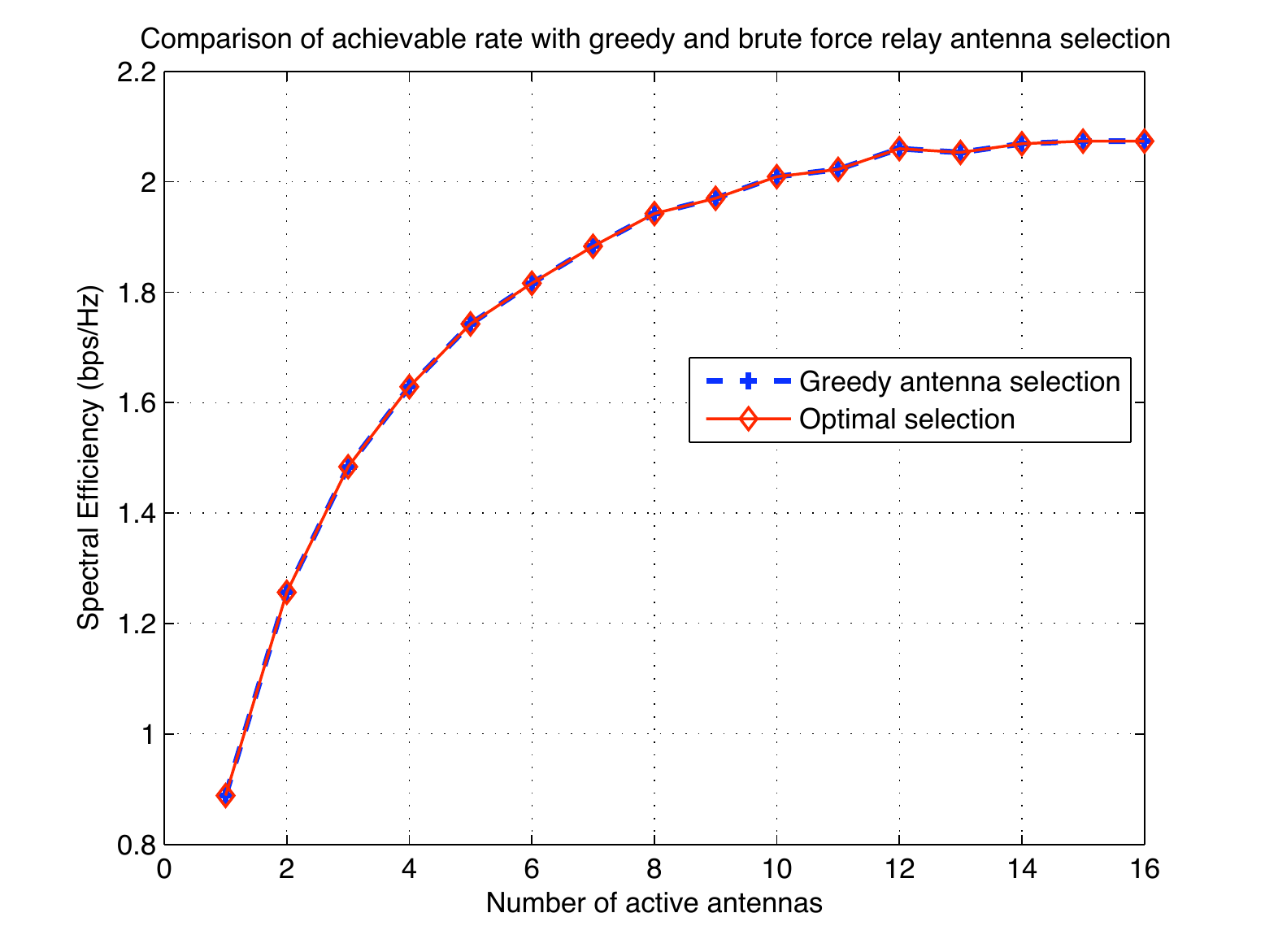}
\caption{Comparison of greedy v/s optimal strategy for relay antenna selection.}
\label{fig:simrelay}
\end{figure}

\section{Conclusion}\label{sec:conc}
In this paper, we used the concept of sub-modular functions to obtain theoretical guarantees on the performance of greedy algorithms for receive antenna selection in point-to-point MIMO channels and relay selection in multiple relay network. The simulated performance of greedy algorithms has been well known in the literature, however, no known theoretical guarantees were available. 
There are many other related challenging antenna selection problems such as: transmit antenna selection in point-to-point MIMO channels for maximizing mutual information, and maximizing the minimum eigen-value \cite{HeathAntSel2001}, relay antenna selection with multiple antennas at the source and the destination with and without channel state information. It is easy to construct numerical examples to show that none of these selection problems, except the two cases considered in this paper, however are sub-modular/modular functions and finding  theoretical bounds on their performance remains an open problem. The numerical examples are not presented here for brevity.

\bibliographystyle{../../IEEEtran}
\bibliography{../../IEEEabrv,../../Research}

\end{document}